\documentclass[a4paper,oneside,12pt]{article}

\usepackage{amsmath,amsthm,amsfonts,amscd,amssymb}
\usepackage{geometry}
\usepackage[english]{babel}
\usepackage[utf8]{inputenc}
\usepackage[T1]{fontenc}
\usepackage{graphicx}
\usepackage{float}
\usepackage{pstricks}
\usepackage{bbm}
\usepackage{mathtools}
\usepackage{nicefrac}
\usepackage{calrsfs}
\usepackage{framed}
\usepackage{hyperref}
\usepackage{enumitem}
\mathtoolsset{showonlyrefs}

\newcommand{\ep}{\varepsilon}
\newcommand{\eps}{\varepsilon}
\newcommand{\bbP}{{\mathbb P}}

\newcommand{\calD}{{\mathcal D}}

\newcommand{\calL}{{\mathcal L}}
\newcommand{\bbZ}{{\mathbb Z}}

\newtheorem{theorem}{Theorem}[section]

\newtheorem{corollary}[theorem]{Corollary}
\newtheorem{lemma}[theorem]{Lemma}

\newtheorem{definition}[theorem]{Definition}

\newcommand{\im}[1]{}

\title{Conformal invariance of the planar four-state Potts model and its random-cluster representation}
\date{6 October 2026}
\author{Emile Averous\footnote{Universit\'e de Gen\`eve, \url{emile.averous@unige.ch}}, Hugo Duminil-Copin\footnote{Universit\'e de Gen\`eve, \url{hugo.duminil@unige.ch}, Institut des Hautes \'Etudes Scientifiques,
\url{duminil@ihes.fr}}, Tiancheng He\footnote{Universit\'e de Gen\`eve, \url{tiancheng.he@unige.ch}}, \\ Piet Lammers\footnote{CNRS, Sorbonne Universit\'e, LPSM,
\url{piet.lammers@cnrs.fr}} and Ioan Manolescu\footnote{Universit\'e de Fribourg,
\url{ioan.manolescu@unifr.ch}}}

\begin{document}
\maketitle

\begin{abstract}
For the critical random-cluster model on $\mathbb Z^2$ with cluster-weight $q=4$, we prove the convergence of loops to CLE(4)
in the scaling limit. Conformal invariance of the four-state Potts model is then deduced via the Edwards--Sokal coupling. 
\end{abstract}

\begin{figure}[H]
    \centering
    \includegraphics[width=0.6\textwidth]{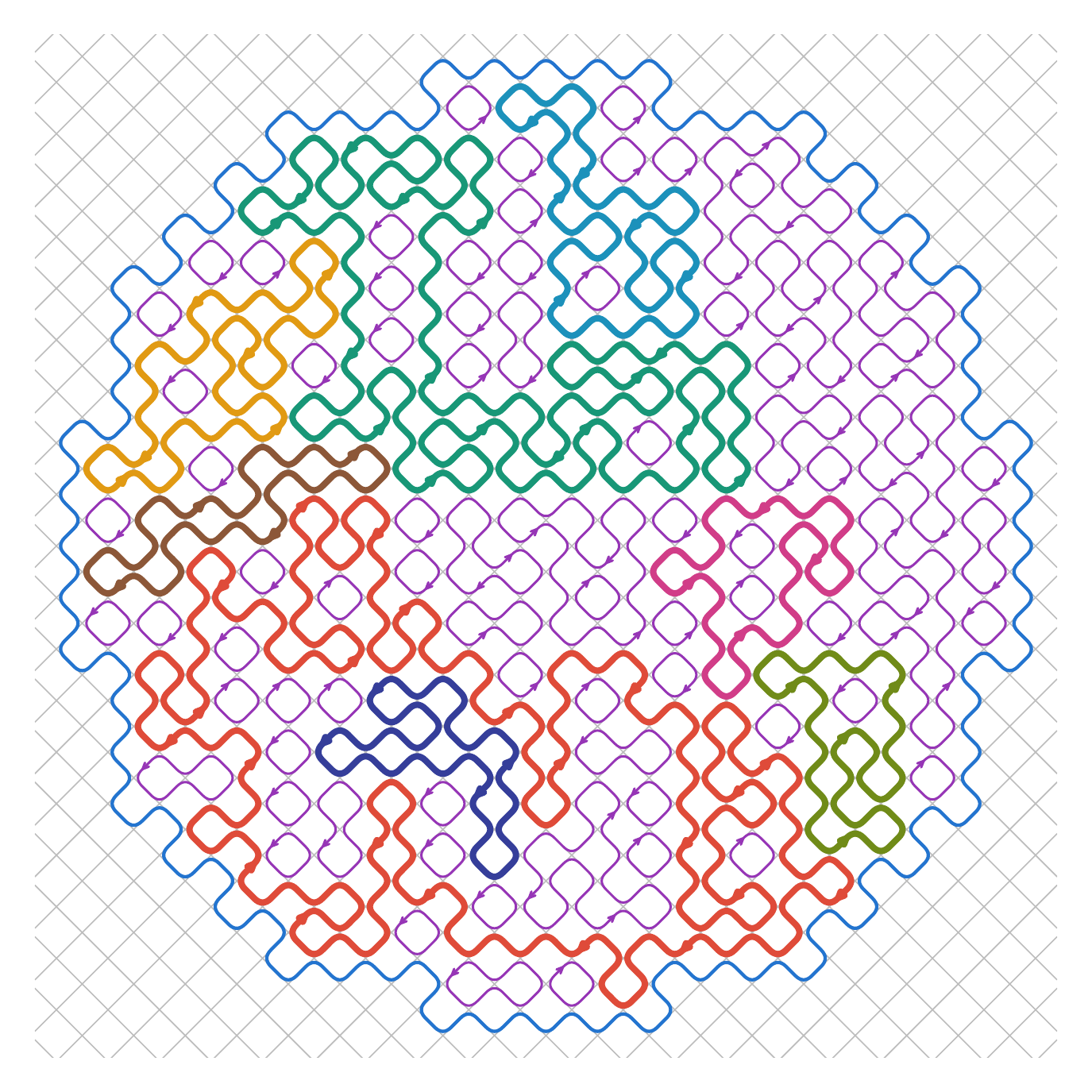}
    \caption{\footnotesize\textbf{The critical random-cluster model at $q=4$.}
    One may view the model as follows:
   fix a deterministic boundary loop on the medial lattice of $\bbZ^2$
   (a scaled, rotated version of the square lattice),
    and sample, uniformly at random, a collection of directed non-intersecting
    loops inside the domain, occupying all edges of the medial lattice.
    In this simulation, we highlighted the eight largest loops (in terms of enclosed area).
    Our result asserts that, in the scaling limit, the loops' law converges to CLE(4).
    }
    \label{fig:front_page}
\end{figure}


\newpage

\section{Introduction}

\subsection{Motivation}

Conformal Field Theory (CFT) has been a cornerstone of theoretical physics in the twentieth century \cite{BPZ84a,BPZ84b}. The introduction of the Schramm--Loewner Evolution (SLE) \cite{Sch00,Law05} and Conformal Loop Ensembles (CLE) \cite{She09,SheWer12} (see also \cite{MS16,MS20}), together with major advances in the mathematical foundations of CFT \cite{Dav+16,RV14,GRV19}, has profoundly reshaped our understanding of conformally invariant objects and two-dimensional CFTs.

Despite this progress, rigorous proofs of conformal invariance for lattice models remain remarkably scarce. Among the principal examples are: Bernoulli site percolation on the triangular lattice \cite{Smi01,KhrSmi21,CamNew07,CamNew06} (results range from Cardy's formula to convergence of interfaces and loop ensembles toward SLE$_6$ and CLE$_6$); the Ising and FK-Ising models \cite{Smi10,CheSmi12,HonSmi13,CheHonIzy15,CheDumHon12a,BenHon16} (major breakthroughs include the analysis of fermionic observables, energy and spin fields, and convergence toward SLE and CLE with parameters $3$ and $16/3$); uniform spanning trees \cite{LawSchWer04}; dimer models \cite{Ken00,Ken01}; and level lines of the discrete Gaussian free field \cite{SchShe09}.

More recently, progress towards a more systematic analysis of scaling limits of lattice models was achieved with the proof of conformal invariance of the six-vertex model for a continuum of parameters; see \cite{DumKozLamMan26} for the full-plane result and \cite{ADCHLM26} for the corresponding result in simply connected domains.
In the present paper, we build on \cite{ADCHLM26} to establish conformal invariance of the random-cluster model with cluster weight $q=4$ on the square lattice, as well as of the associated four-state Potts model. This represents a special case of the Ashkin-Teller model, which we plan to treat in the subsequent paper \cite{AveDumLamMan26}, generalizing the arguments in \cite{DumLisQia21,DumLisQia25}.

\subsection{Definition of the random-cluster model} Consider a finite subgraph $G=(V,E)$ of $\mathbb Z^2$. 
A percolation configuration $\omega$ on $G$ is an element of $\{0,1\}^E$. We say that an edge $e$ is \emph{open} if $\omega(e)=1$ and that it is  \emph{closed} if  $\omega(e)=0$. We can see a configuration $\omega$ as a subgraph of $G$ with the vertex-set $V$ and the edge-set $\{e\in E:\omega(e)=1\}$. 
A \emph{cluster} is a connected component of $\omega$ (possibly including isolated vertices). 
\begin{definition}\label{def.RC_model}
    The random cluster measure on $G$ with edge-weight $p\in(0,1)$, cluster-weight $q>0$ and free boundary conditions is given by
    \begin{equation*}
        \phi^0_{G,p,q}[\omega]:=\frac{1}{Z^0(G,p,q)} \left(\frac{p}{1-p}\right)^{|\omega|}q^{k\left(\omega\right)},
    \end{equation*}
    where $|\omega|:=\sum_{e\in E}\omega(e)$, $k(\omega)$ is the number of connected components of the graph $\omega$ that intersect $V$, and $Z^0(G,p,q)$ is the \emph{partition function} making $\phi^0_{G,p,q}$ a probability measure.
\end{definition}

Consider the dual graph $(\mathbb Z^2)^*:=(\tfrac12,\tfrac12)+\mathbb Z^2$. Note that every edge $e$ of $\mathbb Z^2$ has a dual-edge $e^*$ of $(\mathbb Z^2)^*$ crossing it in its middle. For a percolation configuration $\omega$ on $\mathbb Z^2$, let $\omega^*$ be the configuration on $(\mathbb Z^2)^*$ defined by the formula $\omega^*(e^*)=1-\omega(e)$ for every edge $e$. In words, a dual-edge is open if the corresponding edge of $\mathbb Z^2$ is closed, and vice-versa.

Consider now the medial graph $(\mathbb Z^2)^\diamond$ of the square lattice obtained by putting a medial vertex at the middle of each edge of $\mathbb Z^2$, and medial edges between nearest medial vertices. The graph is a rescaled (by a factor $1/\sqrt2$) and rotated (by an angle $\pi/4$) version of the square lattice. For a configuration $\omega$ on $\mathbb Z^2$, there is a unique way to introduce a loop configuration $\overline\omega$ on $(\mathbb Z^2)^\diamond$ in such a way that loops are drawn in the medial lattice, they turn at each medial vertex so that they do not cross any open edge of $\omega$ or $\omega^*$, and all medial edges are part of exactly one loop; see Figure~\ref{fig:primal dual medial}.

\begin{figure}
    \centering
    \includegraphics[width=0.6\textwidth]{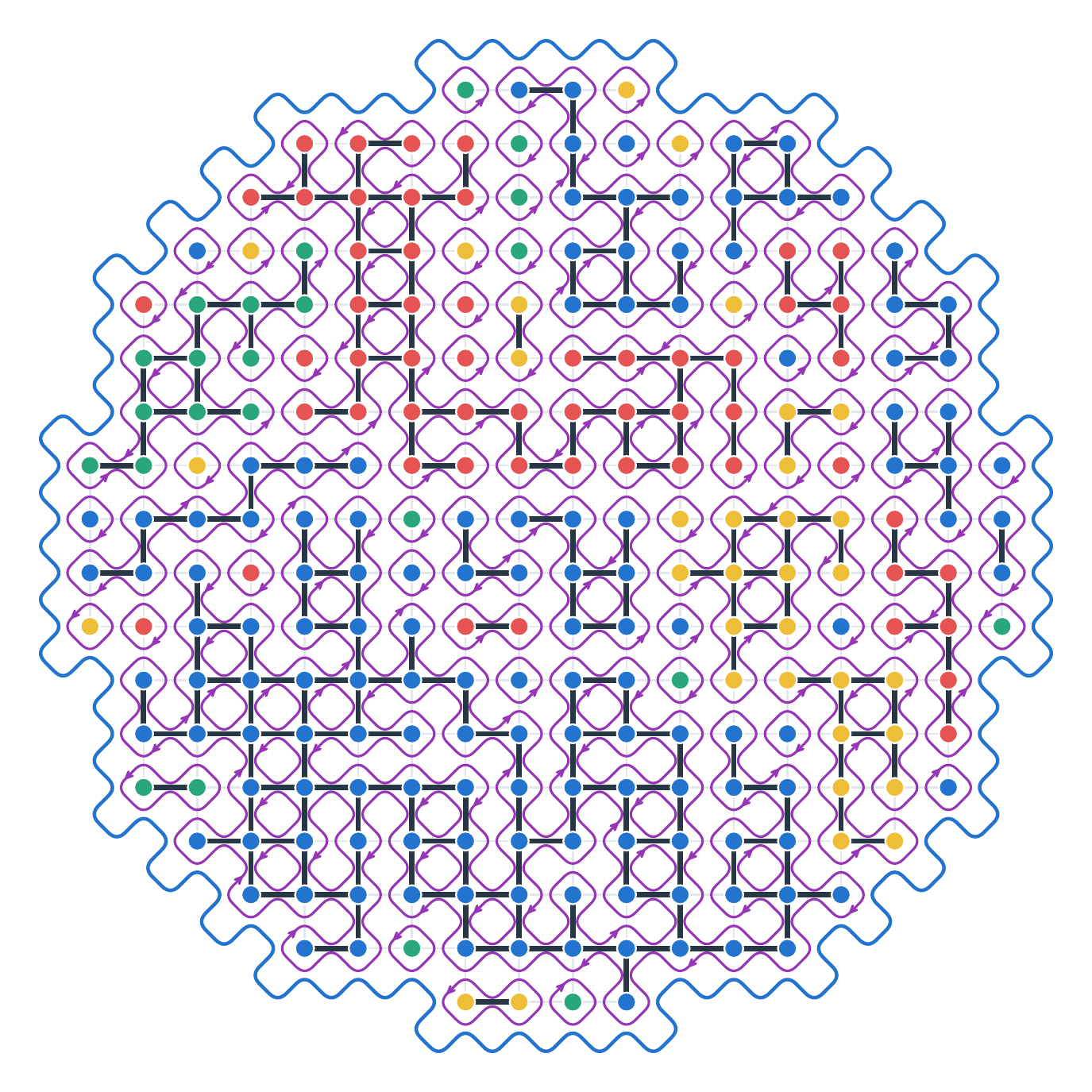}
    \caption{\textbf{Coupling of several models.}
    The \textbf{random-cluster model} is a percolation model on the square lattice $\mathbb Z^2$.
    This simulation has free boundary conditions.
    The \textbf{Potts model} is obtained by colouring its clusters uniformly at random with four colours.
    The \textbf{loops} live on the medial lattice and separate the primal clusters from the dual clusters (the dual clusters are not drawn).
    They are given an \emph{orientation} by flipping a fair coin for each loop.
    Notice that the random-cluster model and the \emph{unoriented} loops encode precisely the same information.}
    \label{fig:primal dual medial}
\end{figure}

In  what follows, we call  {\em discrete domain} any subgraph $G$ of $\mathbb Z^2$ obtained by taking all the vertices (of $\mathbb Z^2$) within a loop $\ell$ on the medial lattice $(\mathbb Z^2)^\diamond$ turning at every vertex, and all edges (of $\mathbb Z^2$) between these vertices not crossing the loop $\ell$. A discrete domain can be seen as a union of simply connected open sets of $\mathbb R^2$ by considering the points enclosed by the loop\footnote{The loops are considered as making small roundish turns at every medial vertex (cf.~Figure~\ref{fig:front_page}).}. The boundary $\partial G$ of $G$ is the set of vertices in $V$ that correspond to (medial) faces of $(\mathbb Z^2)^\diamond$ bordering the loop $\ell$.

We consider a percolation configuration $\omega$ on a finite graph $G=(V,E)$ with free boundary conditions as a configuration on the whole lattice $\mathbb Z^2$ by declaring all edges which are not in $E$ to be closed.  As a consequence, we may define $\overline\omega$ by considering the full-plane configuration and then  removing all the small loops surrounding vertices which are not in $V$.  With this definition, the configuration $\overline\omega$ is composed entirely of loops.

For more details on the random-cluster model, see e.g.~\cite{Gri06,Dum17a}.
\medbreak
{\em In this article, we focus on the cluster weight $q=4$. In such a case, it has been proved that the model undergoes a continuous phase transition at the self-dual point $p_c=2/3$ \cite{DumSidTas16}. In all that follows, we set $p:=p_c=2/3$.}

\subsection{Conformal invariance of the random-cluster model at \texorpdfstring{$q=4$}{q=4}}

The object of interest in the continuum will be the nested Conformal Loop Ensemble (CLE) with parameter $\kappa$ equal to 4. We refer to \cite{She09,SheWer12} for its definition. Here, we simply recall that the CLE(4) in a simply connected domain $\calD$ is a random collection of nested, non intersecting and non self-crossing continuous loops in  $\calD$ which is conformally invariant, in the sense that CLE(4) in $\Psi(\calD)$, where $\Psi$ is a conformal map, is the image by $\Psi$ of CLE(4) in $\calD$.

We say that a family of discrete domains 
$
\calD_\delta
$ on $\delta\mathbb Z^2$ converges to
$\calD
$
in the {\em Carath\'eodory sense} if there exist conformal bijections
$
f_\delta:\calD_\delta\longrightarrow\mathbb D
$ and $
f:\calD\longrightarrow\mathbb D
$
such that every compact subset of $\calD$ is contained in $\calD_\delta$
for all sufficiently small $\delta$, and $f_\delta$ converges to $f$
uniformly on every compact subset of $\calD$. 

\begin{theorem}[CLE(4) convergence of macroscopic loops (Figure~\ref{fig:front_page})]\label{thm:main}
Consider a simply connected domain $\calD$ and let $(\calD_\delta)$ be a family of discrete domains on $\delta\mathbb Z^2$ converging to $\calD$ in the Carath\'eodory sense. Let $\overline\omega_\delta$ be the loop configuration obtained from the random-cluster model on $\calD_\delta$ with cluster-weight $q=4$ and edge-weight $p=2/3$, with free boundary conditions. Then, $\overline\omega_\delta$ converges weakly (for the topology induced by the Camia-Newman distance\footnote{For a simply connected domain $\Omega$, let~$\mathfrak C=\mathfrak C(\Omega)$ be the collection of locally finite families~$\mathcal F$ of non-self-crossing loops in $\Omega$ that do not intersect each other. Define the metric on~$\mathfrak C$, 
\[
d_{\bf CN}(\mathcal F,\mathcal F')\le \varepsilon\ \Longleftrightarrow\ \Big(\begin{array}{c} \forall \gamma\in \mathcal F\text{ with }\gamma\subset B(0,1/\ep),\exists \gamma'\in \mathcal F
',d(\gamma,\gamma')\le \ep\\
\text{ and similarly when exchanging~$\mathcal F'$ and~$\mathcal F$}\end{array}\Big),
\]
where, for two loops~$\gamma_1$ and~$\gamma_2$, we set 
\begin{align}\label{eq:loop_dist}
	d(\gamma_1,\gamma_2):=\inf\sup_{t\in\mathbb S^1}|\gamma_1(t)-\gamma_2(t)|,
\end{align}
with the infimum running over all continuous one-to-one parametrizations of the loops~$\gamma_1$ and~$\gamma_2$ by~$\mathbb S^1$.

The distance~$d_{\bf CN}$ induces a distance on the space of probability measures on~$\mathfrak C$, which we also denote by~$d_{\bf CN}$: for~$\phi$ and~$\phi'$ two measures on~$\mathfrak C$, let
\begin{align}\label{eq:dis_law}
d_{\bf CN}(\phi,\phi') := \inf\big\{\exists \text{ coupling~$\bbP$ of~$\phi$ and~$\phi'$ with }	\mathbb P[d_{\bf CN}(\overline\omega,\overline\omega')> \ep] <\ep\big\}.
\end{align}}) to the CLE(4) in $\calD$. 
\end{theorem}

\im{Are all domains considered here bounded? The discrete ones are, but what about the continuous ones. Is there any problem with the CN topology? }

Using the universality result of \cite{DumKozKraManOul26}, the theorem extends to every rectangular isoradial graph. We chose to state it in the case of the square lattice for simplicity. Previous results concerning the random-cluster model with cluster-weight $q=4$ include the derivation of the critical exponents \cite{Che26a}, as well as certain correlations in full plane \cite{CheXia26}. Both use the GFF convergence of \cite{DumKozLamMan26}.

The following corollary, based on our main result combined with mixing properties of the model, states that the multi-point connectivity correlations are conformally invariant. The exact expression of these correlations can be obtained in terms of CLE(4) correlations, but we do not include it here. 
Below, for a point $x\in\mathbb R^2$, let $[x]_\delta$ be the vertex of $\delta\mathbb Z^2$ closest to $x$ (in case of a tie we choose one of the closest vertices according to an arbitrary rule). 

\begin{corollary}[Conformal invariance of point-to-point correlations]\label{cor:conformal invariance point-to-point}
There exists $\kappa(\delta)$ such that $\lim_{\delta\rightarrow0}\frac{\log(\kappa(\delta))}{\log(1/\delta)}=1/8$ and the following holds. Consider a simply connected domain $\calD$ and $n\ge1$. For any partition $P=P_1\sqcup\dots\sqcup P_s$ of $\{1,\dots,n\}$, there exists a map $\Psi_{\calD}(P;\cdot)$ such that for every $x_1,\dots,x_n\in \calD$, 
\begin{equation}
\lim_{\delta\rightarrow0}\kappa(\delta)^{n}\phi_{\calD_\delta,2/3,4}^0[\mathcal E_\delta(P;x_1,\dots,x_n)]=\Psi_\calD(P;x_1,\dots,x_n),\end{equation}
where $\mathcal E_\delta(P;x_1,\dots,x_n)$ denotes the event that $[x_i]_\delta$ and $[x_j]_\delta$ are connected in $\omega_\delta$ if and only if $i$ and $j$ are in the same $P_t$ for some $1\le t\le s$.  

Furthermore, the family of functions $(\Psi_\calD(P;\cdot))_\calD$ is conformally covariant in the following sense: for every conformal map $f:\calD\rightarrow \mathbb C$ and every $x_1,\dots,x_n\in \calD$,
\begin{equation}
\Psi_\calD(P;x_1,\dots,x_n)=|f'(x_1)|^{1/8}\cdots|f'(x_n)|^{1/8}\Psi_{f(\calD)}(P;f(x_1),\dots,f(x_n)).\end{equation}
\end{corollary}
%
%
%

\subsection{Conformal invariance of the  Potts model}

 For a finite subgraph $G=(V,E)$ of $\bbZ^2$, attribute a {\em spin} variable $\sigma_x$ belonging to $\{1,\dots,q\}$ to each vertex $x\in V$. A
 {\em spin configuration} $\sigma=(\sigma_x:x\in V)\in\{1,\dots,q\}^{V}$ is given by the collection of all the spins.
  Introduce the Hamiltonian of $\sigma$ defined by
$$H_G(\sigma):=-\sum_{xy\in E}\,\mathbbm{1}[\sigma_x =\sigma_y].$$
The {\em Gibbs measure on $G$ at inverse temperature $\beta\ge0$ with free boundary conditions} is defined by the formula
\begin{equation}\label{eq:Gibbs}\mu_{G,\beta,q}[\sigma]:=\frac{1}{Z(G,\beta,q)}\exp[-\beta H_{G}(\sigma)],\end{equation}
for every $\sigma\in\{1,\dots,q\}^{V}$, and $Z(G,\beta,q)$ is the partition function making the previous measure a probability measure. 

This model was introduced as a generalization of the Ising model to more than two possible spins by Potts in 1952 \cite{Pot52} following a suggestion of his adviser Domb. While the model received little attention early on, it became an object of great interest in the last forty years. Since then, mathematicians and physicists have been studying it intensively, and a lot is known on its rich behavior.

The Potts model with $q$ states is related to the random-cluster model with cluster-weight $q$ via the Edwards--Sokal coupling ; see e.g.~\cite{Gri06,Dum17a}. The Edwards--Sokal coupling has many applications. Here, we mention the simplest statement, and leave more evolved statements to the reader.
Below, we focus on the case $q=4$. It is easy to deduce from the critical point of the random-cluster model that the $4$-state Potts model undergoes a continuous phase transition at $\beta_c=\log(3)$. Recall that $\kappa(\delta)$ is given by Corollary~\ref{cor:conformal invariance point-to-point} and satisfies $\lim_{\delta\rightarrow0}\frac{\log(\kappa(\delta))}{\log(1/\delta)}=1/8$.

\begin{corollary}\label{cor:conformal invariance Potts}
Consider a simply connected domain $\calD$.
Fix a tuple of spins $\tau\in\{1,\dots,q\}^n$.
There exists $\widetilde\Psi_{\calD}(\tau;\cdot)$ such that for every $x_1,\dots,x_n\in \calD$,\begin{equation}
\lim_{\delta\rightarrow0}\kappa(\delta)^{n}\mu_{\calD_\delta,\beta_c,4}\!\left[\prod_{i=1}^n\left(\mathbbm{1}[\sigma_{[x_i]_\delta}=\tau_i]-\tfrac14\right)\right]=\widetilde \Psi_\calD(\tau;x_1,\dots,x_n).
\end{equation}
Furthermore, the family of functions $(\Psi_\calD(P;\cdot))_\calD$ is conformally covariant in the following sense: for every conformal map $f:\calD\rightarrow \mathbb C$ and every $x_1,\dots,x_n\in \calD$,
\begin{equation}
\widetilde\Psi_\calD(P;x_1,\dots,x_n)=|f'(x_1)|^{1/8}\cdots|f'(x_n)|^{1/8}\widetilde\Psi_{f(\calD)}(P;f(x_1),\dots,f(x_n)).\end{equation}

\end{corollary}

\section{Proof of Theorem~\ref{thm:main}}

The proof relies on a coupling between the random-cluster model and the height function of the six-vertex model. We start by defining the six-vertex model. Then, we recall the convergence result obtained in \cite{ADCHLM26}. Finally, we prove our result. 

\begin{figure}
    \centering
    \includegraphics[width=0.6\textwidth]{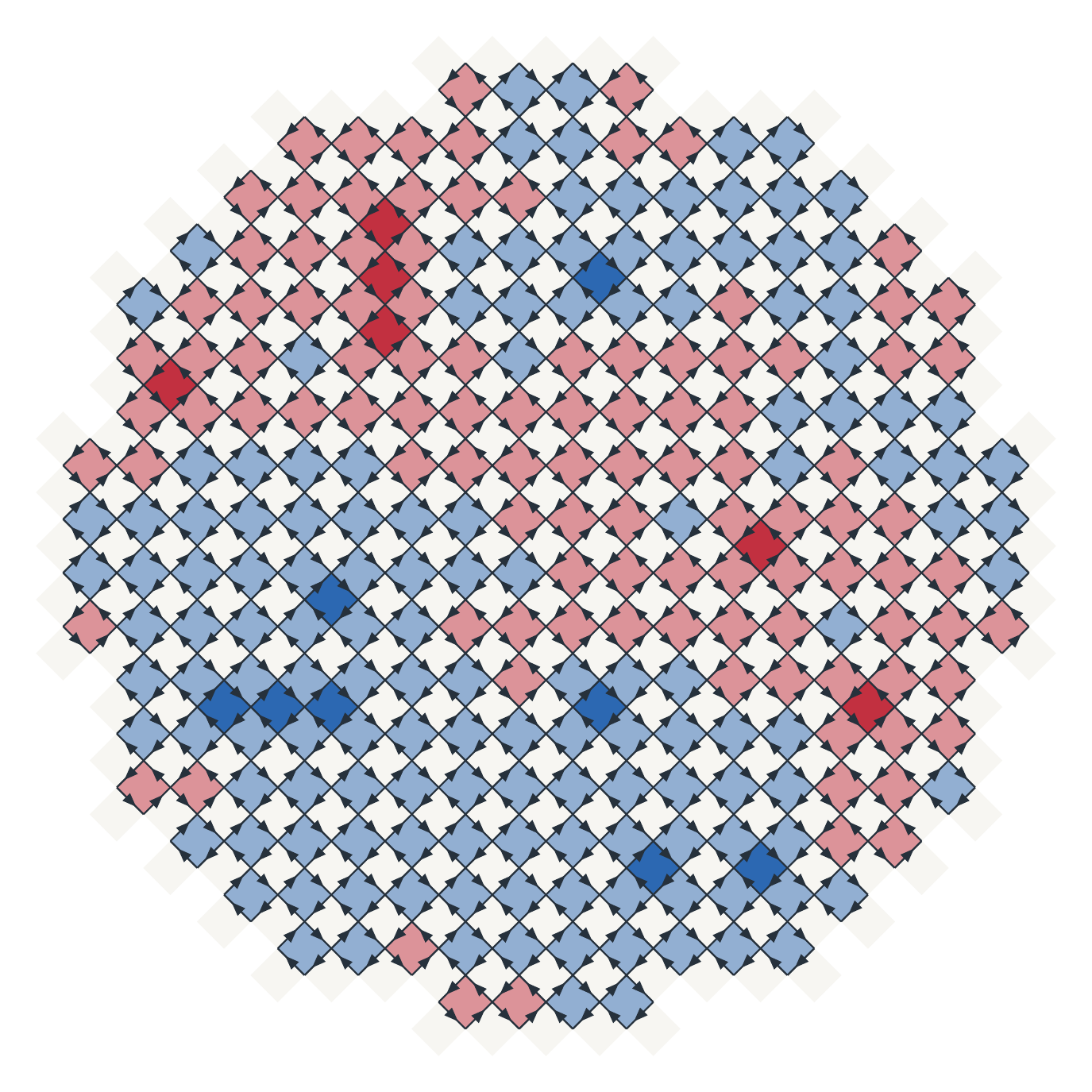}
    \caption{The \textbf{six-vertex configuration} is obtained from the \emph{oriented} loops
    by ``forgetting'' the way that the oriented edges are wired up into loops;
    see Figure~\ref{fig:primal dual medial} for the same simulation \emph{with} the orientations.}
    \label{fig:six_vertex_bkw}
\end{figure}

\subsection{The six-vertex model }

The six-vertex model is a probability measure on orientations of the edges of a graph \cite{Bax82,LieWu80}. A six-vertex configuration $\omega^\times$ on a subgraph of the medial lattice $(\mathbb Z^2)^\diamond$ assigns an orientation to each edge, subject to the \emph{ice rule}: at every vertex of degree four, the number of incoming edges equals the number of outgoing edges. In particular, there are precisely six local arrow configurations compatible with this constraint, as depicted in Figure \ref{fig:1}. The six local configurations are grouped into three pairs, to which we assign positive weights $\mathbf a$, $\mathbf b$, and $\mathbf c$, referred to below as {\em Boltzmann weights}.

The six-vertex model is often considered with periodic boundary conditions.
In this paper, however, we work exclusively in finite domains with fixed boundary conditions. For a discrete domain $G$, let $G^\diamond$ be the collection  $F$ of faces of $(\delta\mathbb Z^2)^\diamond$ which are enclosed by the medial loop defining the boundary of $G$. Note that the degree of vertices in $G^\diamond$ is two or four. 

 Every six-vertex configuration $\omega^\times$ on $G^\diamond$ naturally gives rise to an integer-valued \emph{height function} $h$ defined on $F$. It is piecewise constant on each face and its increments between neighboring faces are determined by the orientation of the edge separating them: when crossing an edge of $\omega^\times$, the height increases by $1$ if the crossed arrow points to the right and decreases by $1$ otherwise. 

\begin{figure}
\centering
\includegraphics[width=0.8\linewidth]{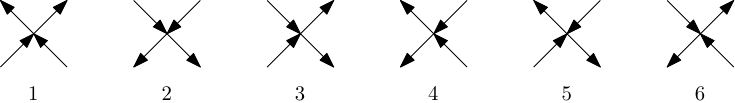}
\caption{The six possible local arrow configurations around a square lattice vertex of the six-vertex model.
These local configurations get the weights $\mathbf a$, $\mathbf a$, $\mathbf b$, $\mathbf b$, $\mathbf c$, $\mathbf c$ respectively.
The $\mathbf c$-type vertices are special; they correspond to saddle points of the height function.}
\label{fig:1}
\end{figure}

The ice rule ensures that these local increments are compatible around every vertex and can therefore be integrated to define a height function on $F$. The height function is determined only up to a global additive constant, whereas its gradient is uniquely determined by $\omega^\times$. Conversely, the arrow configuration can be recovered from the height increments between neighboring faces. Thus, modulo an additive constant, the height function and the six-vertex configuration encode the same information.

Consider a discrete domain $G$ and recall the definition of $\partial G$, i.e.~the vertices of $\mathbb Z^2$ (and therefore the associated medial faces) that are adjacent to the medial loop defining the boundary of $G$.
The six-vertex measure on $G^\diamond$ with \emph{Dirichlet boundary condition 0} is defined by
\begin{equation}
\label{eq:6V_def}
\bbP^0_{G,\mathbf a,\mathbf b,\mathbf c}[\omega^\times]
:=
\frac{1}{Z^0_{G,\mathbf a,\mathbf b,\mathbf c}}
\mathbf a^{N_1+N_2}
\mathbf b^{N_3+N_4}
\mathbf c^{N_5+N_6}
\mathbbm{1}_{h_{|\partial G}=0},
\end{equation}
where $N_i=N_i(\omega^\times)$ denotes the number of degree-four vertices whose local arrow configuration is of type $i$. The normalizing constant $Z^0_{G,\mathbf a,\mathbf b,\mathbf c}$ is the corresponding \emph{partition function}, chosen so that \eqref{eq:6V_def} defines a probability measure. 
\medbreak
{\em Below, we restrict to the parameters $\mathbf a=\mathbf b=1$ and $\mathbf c=2$.}

\subsection{Scaling limit of the six-vertex height function}

Consider a simply connected domain $\calD$. Let
$ G_{\calD}:D_2(\calD):=\{(u,v)\in \calD^2:u\ne v\}\longrightarrow\mathbb R
$
denote the Dirichlet Green function in $\calD$ defined as follows: for each fixed $y\in\calD$, the function $x\mapsto G_{\calD}(x,y)$ is harmonic on $\calD\setminus\{y\}$, vanishes as $x$ approaches $\partial\calD$, and satisfies
\begin{equation}
    G_{\calD}(x,y)
    =
    -\frac{1}{2\pi}\log|x-y|+O(1)
    \qquad\text{as }x\to y.
\end{equation}

By a \emph{finite-Dirichlet-energy generalized test function} we mean a finite signed measure $\varphi$ supported in $\calD$ such that
\begin{equation}
    \iint_{\calD\times\calD}
    G_{\calD}(u,v)\,{\rm d}\varphi(u)\,{\rm d}\varphi(v)
    <\infty.
\end{equation}
The Gaussian Free Field (GFF) is a centered Gaussian generalized function $\Gamma$ with covariance kernel $G_{\calD}$; see \cite{BerPow25,She07}. More precisely, it is a random distribution, in the sense that it can be tested against a finite-Dirichlet-energy generalized test function $\varphi$ by setting
\begin{equation}
   \langle\Gamma,\varphi\rangle:= \int_{\calD}
   \Gamma(u){\rm d}\varphi(u).
\end{equation}
Then,
for every finite family
$
   \varphi=(\varphi_1,\ldots,\varphi_n)
$ of finite-Dirichlet-energy generalized test functions, 
\begin{equation}
    \bigl(\langle\Gamma,\varphi_i\rangle\bigr)_{i=1}^n
    \sim
    \mathcal N\bigl(0,\Sigma(\varphi)\bigr),
\end{equation} 
where 
$\Sigma(\varphi)$ is the  $n\times n$ matrix defined by
\begin{equation}
    \Sigma(\varphi)_{ij}
    :=
    \iint_{\calD\times\calD}
    G_{\calD}(u,v)\,
    {\rm d}\varphi_i(u)\,
    {\rm d}\varphi_j(v).
\end{equation}

The following convergence result for the six-vertex height function was obtained in \cite{ADCHLM26} and is the basis of our proof.

\begin{theorem}[GFF convergence of the six-vertex height function \cite{ADCHLM26}]\label{thm:even}
Consider the six-vertex model with weights $\mathbf a=\mathbf b=1$ and $\mathbf c=2$.  Let $\calD$ be a simply connected domain. For a family $(\calD_\delta)$ of discrete domains converging to
$\calD$
in the Carath\'eodory sense, the law of the height function $h_\delta\sim\bbP_{\calD_\delta^\diamond,1,1,2}^0$ of the six-vertex model associated with $\omega\sim\phi_{\calD_\delta,2/3,4}^0$,  converges to the GFF $\sigma\cdot\Gamma$ with zero boundary conditions on $\calD$, and variance $\sigma^2=\frac2\pi$.\end{theorem}

\subsection{The Baxter--Kelland--Wu correspondence}

 There is a formal correspondence, called the Baxter--Kelland--Wu (BKW) correspondence, between the loop configurations and the arrow configurations. We refer to \cite{BKW76} or the more recent expositions \cite{Dub11,DumGagHarManTas20,AveCheDumHeKraManXia26} for details.
In the case of $\mathbf a=\mathbf b=1$ and $\mathbf c=2$, it is a probabilistic coupling between the six-vertex model and the random-cluster model, that we now state.  
 
 Consider a discrete domain $G$. Consider the uniform measure $\vec{\mathbb  P}_G$ on oriented loop configurations $\vec\omega$ on $G^\diamond$, i.e.~loop configurations $\overline\omega$ for which each loop receives an orientation, clockwise or counterclockwise. Then, the BKW coupling states the following two observations.
 \begin{itemize}
 \item If $\overline\pi$ denotes the projection from oriented loop configurations to loop configurations consisting in forgetting orientations of loops, then
 \begin{equation}\phi_{G,2/3,4}^0=\overline\pi_*\vec{\mathbb  P}_G.\end{equation}
 (See Figure~\ref{fig:primal dual medial}.)
 \item If $\vec\pi$ denotes the projection from oriented loop configurations to arrow configurations consisting in forgetting the loops and simply keeping the orientations of edges induced by the orientations of the loops\footnote{Note that when considering the BKW coupling starting from a configuration $\omega\sim\phi_{G,2/3,4}^0$, the corresponding boundary conditions of the six-vertex height function are constant in the following sense: the height function restricted to medial faces corresponding to vertices of $\partial G$, is constant.}, and setting the height function to be zero on $\partial G$, then
 \begin{equation}\mathbb P^0_{G,1,1,2}=\vec\pi_*\vec{\mathbb  P}_G.\end{equation}
 (See Figure~\ref{fig:six_vertex_bkw}.)
 \end{itemize}

\subsection{Proof of Theorem~\ref{thm:main}}

The next lemma states a natural decomposition coming from the coupling between the height function and the loop configuration.

\begin{figure}[th!]
    \centering
    \includegraphics[width=0.55\textwidth]{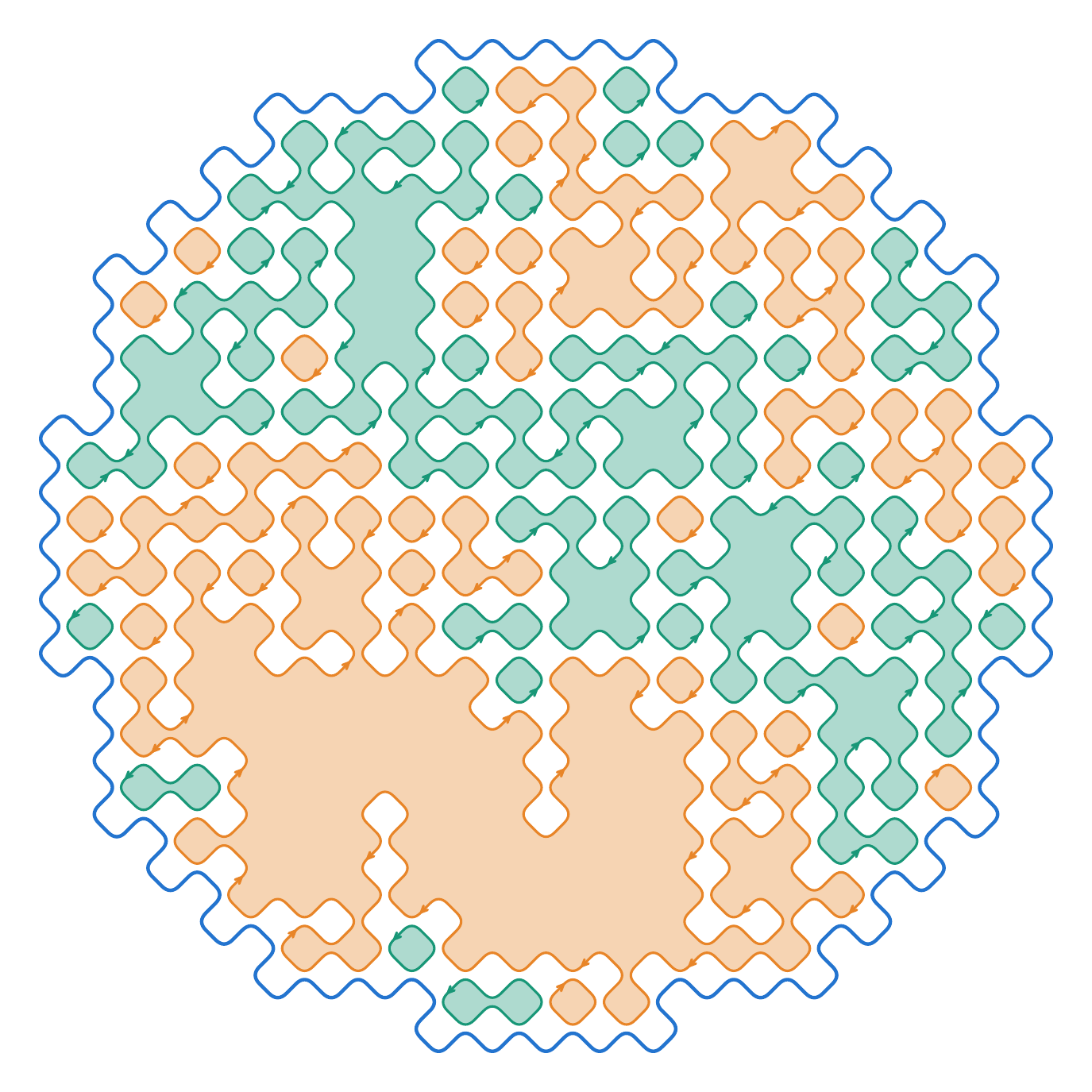}
    \caption{\textbf{A discrete version of the two-valued local set of Lemma~\ref{lemma:discrete_two_valued_set}:}
    The decomposition
    is obtained by exploring the outermost loops in the domain.
    The orientations of the loops are independent fair coin flips.
    By the Markov property, the law inside each loop is given by an independent copy of the height function with zero boundary conditions,
    plus the coin flip coming from the orientation of that loop.\\
    {\color{white}---}In the proof of Theorem~\ref{thm:main}, we argue
    that this discrete decomposition commutes with the scaling limit
    and yields the appropriate two-valued local set of the GFF $\sigma\cdot\Gamma$,
     which enables
    us to identify the scaling limit of the outermost loops as CLE(4).}
    \label{fig:discrete_two_valued_set}
\end{figure}

\begin{lemma}[Discrete two-valued local set (cf. Figure~\ref{fig:discrete_two_valued_set})]
    \label{lemma:discrete_two_valued_set}
Consider a discrete domain $G$, and let $h_G$ be the six-vertex height function in $G$ with zero boundary conditions and variance $\sigma^2=\frac 2\pi$, let $\mathcal L=\mathcal L(\overline\omega)$ be the collection of exterior-most loops in $\overline \omega$, and let $\calD(\ell)$ be the discrete domain obtained by considering all the medial faces fully enclosed by $\ell$. We have the following identity in law:
\begin{equation}\label{eq:discrete decomposition}
h_G=\sum_{\ell\in\mathcal L(\overline\omega)}\bigl(\xi_\ell\mathbbm{1}_{\calD(\ell)}+h_{\calD(\ell)}\bigr),
\end{equation}
where the $\xi_\ell$ are iid $\pm1$ Bernoulli variables, and the $h_{\calD(\ell)}$ are independent copies of the six-vertex height function in $\calD(\ell)$ with zero boundary conditions.
\end{lemma}

\begin{proof}
The decomposition is immediate by the BKW coupling. 
The variable $\xi_\ell$ is equal to $+1$ if the loop $\ell$ is oriented counterclockwise, and $-1$ if it is oriented clockwise.
\end{proof}

Let us remark that in Lemma \ref{lemma:discrete_two_valued_set}, the jumps of the height function from one side of the loop to another are equal to $\pm 1$. In the corresponding GFF limit, the jumps of the GFF $\sigma\cdot\Gamma$ when traversing a loop of the CLE(4) are $\pm \sigma\cdot 2\sqrt{\frac{\pi}8}=\pm 1$ \cite{SchShe09}, so this discrete decomposition has an exact continuum counterpart. We are now in a position to prove Theorem~\ref{thm:main}.

\begin{proof}[Proof of Theorem~\ref{thm:main}]
It was proved in~\cite{SepWer17} that there is a unique thin local set
with the property that the GFF $\sigma\cdot\Gamma$ has boundary values $\pm1$ on the boundary of
each connected component of the complement of the local set.
Moreover, this thin local set is precisely the collection of outermost loops of a CLE(4).
Below, we argue that the decomposition obtained in Lemma~\ref{lemma:discrete_two_valued_set}
commutes with the limit, and yields precisely this thin local set.
This proves that the limit of the law of the outermost discrete loops is precisely CLE(4).

We start by studying the scaling limit of exterior-most loops. RSW theory \cite{DumSidTas16} implies precompactness of the family of loops. Thanks to this precompactness, it suffices to show that any sub-sequential limit $\calL$ of the family of loops $\calL(\overline\omega_\delta)$ is equal to the exterior-most loops of a CLE(4). 
Consider such a sub-sequential limit $\calL$. Let $\mathcal A$ be the ``gasket'' of the loops in $\calL$, meaning the closure of the connected component in $\calD\setminus \cup_{\ell\in \calL}\ell$ of $\partial\calD$. 

First, note that $\mathcal A$ is thin. Indeed, standard RSW estimates \cite{DumSidTas16} imply that for a ball of size $\eps>0$ at distance $r>\eps$ of $\partial\calD$ to be contained in $\mathcal A$, there should be no loop in $\calL$ surrounding it. This occurs with probability $(\eps/r)^c$ for some constant $c>0$. The proof of thinness then follows from the same computation as in the proof of \cite[Lemma 6.10]{DumLisQia25}. 

Second, we claim that $\mathcal A$ is a local set in the sense of \cite{SepWer17}. Indeed, by the Skorokhod representation theorem assume that the $\omega_\delta$ are coupled so that the $\calL(\overline\omega_\delta)$ converge almost surely to $\calL$. Using precompactness for the height function seen as a random distribution, we may extract a sub-sequential limit in which $h_{\calD_\delta}$ and $(h_{\calD(\ell_\delta)}:\ell_\delta\in\calL(\overline\omega_\delta))$ converge almost surely to $h_\calD$ and $(h_{\calD(\ell)}:\ell\in \calL)$, in the sense that for all loops $\ell\in\calL$, if a sequence of loops $\ell_\delta$ of $\calL(\overline\omega_\delta)$ converges to $\ell$, then $h_{\calD(\ell_\delta)}$ converges to $h_{\calD(\ell)}$ (in the sense of Theorem \ref{thm:even}).
From \eqref{eq:discrete decomposition}, we deduce the following identity in law
\begin{equation}
h_\calD=\sum_{\ell\in\mathcal L}\bigl(\xi_\ell\mathbbm{1}_{\calD(\ell)}+h_{\calD(\ell)}\bigr).
\end{equation}
As a consequence, $\mathcal A$ is indeed a local set. Note that we use uniform thinness (as $\delta$ tends to zero) of the gaskets of $\calL(\overline\omega_\delta)$ (which is obtained by RSW theory \cite{DumSidTas16} in the same way as for $\calL$) to justify that small loops in $\calL(\overline\omega_\delta)$ (i.e.~loops of size $o(1)$) do not contribute in the limit.

Third, we identify the values of the field on the boundary of $\mathcal A$. Inside each component of $\mathcal A^c$, the limit height function is a GFF with boundary conditions $\xi_\ell\in\{-1,+1\}$, where $\ell$ is the loop tracing out the boundary of this component. Note that we do not claim that the boundary conditions in each such component are i.i.d.~since one loop $\ell$ can a priori carve multiple components if it touches itself, but it will suffice for us that the boundary conditions are either $-1$ or $+1$.

By the jump calculation preceding the proof and the characterization above, it follows that $\calL$ is the set of exterior most loops of CLE(4); see~\cite{SepWer17}. 

The nested result readily follows by the Markov property and precompactness.
\end{proof}

\section{Proof of the corollaries}

\begin{proof}[Proof of Corollary~\ref{cor:conformal invariance point-to-point}]
Consider a simply connected domain $\calD$ and $n\ge1$. For any partition $P=P_1\sqcup\dots\sqcup P_s$ of $\{1,\dots,n\}$, for every $x_1,\dots,x_n\in \calD$ and $\eps>0$, the CLE(4) convergence implies that 
\begin{equation}
\lim_{\delta\rightarrow0}\phi_{\calD_\delta,2/3,4}^0[\mathcal E_\delta^\eps(P;x_1,\dots,x_n)]=\bbP_{\rm CLE(4)}(\mathcal E^\eps(P;x_1,\dots,x_n)),\end{equation}
where $\mathcal E_\delta^\eps(P;x_1,\dots,x_n)$ denotes the event that the Euclidean balls of size $\eps$ around  $x_i$ and $x_j$ are connected in $\omega_\delta$ (resp.~the complement of loops in CLE(4))  if and only if $i$ and $j$ are in the same $P_t$.  

The mixing property \cite{Che26a} implies that for every $\eta>0$, there exists $\eps>0$ such that 
\begin{equation}
\Big|\phi_{\calD_\delta,2/3,4}^0[\mathcal E_\delta(P;x_1,\dots,x_n)]-\phi^0_{\delta\mathbb Z^2,2/3,4}[0\leftrightarrow \infty|B_\eps\leftrightarrow\infty]^n \phi_{\calD_\delta,2/3,4}^0[\mathcal E_\delta^\eps(P;x_1,\dots,x_n)]\Big|\le\eta.
\end{equation}
Above, we used the measure $\phi^0_{\delta\mathbb Z^2,2/3,4}[\cdot|B_\eps\leftrightarrow\infty]$ which is the limit (as $R$ tends to infinity) of the measures conditioned on the Euclidean ball of size $\eps$ around 0 being connected to the Euclidean ball of size $R$ around $0$.
The result then follows by choosing a sequence $(\eps_\delta)$ slowly tending to zero as $\delta$ tends to zero, and then setting
\begin{equation}
\kappa(\delta):=\eps_\delta^{-1/8}\phi^0_{\delta\mathbb Z^2,2/3,4}[0\leftrightarrow \infty|B_{\eps_\delta}\leftrightarrow\infty]
\end{equation}
and 
\begin{equation}
\Psi_\calD(P;x_1,\dots,x_n):=\lim_{\eps\rightarrow 0}\eps^{-n/8}\bbP_{\rm CLE(4)}(\mathcal E^\eps(P;x_1,\dots,x_n)).
\end{equation}
The asymptotics of $\kappa(\delta)$ are due to \cite[Theorem~1.2]{Che26a}.
The conformal covariance follows easily from properties of CLE(4).
\end{proof}

\begin{proof}[Proof of Corollary~\ref{cor:conformal invariance Potts}]
This is a direct consequence of the Edwards--Sokal coupling and Corollary~\ref{cor:conformal invariance point-to-point}.
\end{proof}

\paragraph{AI use statement.}
No mathematical content was generated or found with the assistance of artificial intelligence.

\paragraph{Acknowledgments.} The second author was funded by the ERC SPRaPaF. The first, second, third, and fifth authors were funded by the NCCR SwissMap and the Swiss FNS.
The fourth author acknowledges support by the French National Research Agency (ANR),
project number ANR-23-CPJ1-0150-01,
and the project ERC GeStruct.


\end{document}